\documentclass[letterpaper, 10 pt, conference]{ieeeconf}  

\IEEEoverridecommandlockouts                              

\usepackage{graphics} 
\usepackage{epsfig} 
\usepackage{mathptmx} 
\usepackage{times} 
\usepackage{setspace}
\usepackage{amssymb}
\usepackage{amsmath}
\usepackage{amsfonts}
\usepackage{bm}

\newcommand{\Pb}{{\mathbb{P}}}

\newcommand{\Eb}{{\mathbb{E}}}
\newcommand{\Zb}{{\mathbb{Z}}}

\newcommand{\Fc}{{\mathcal{F}}}
\newcommand{\Rb}{{\mathbb{R}}}

\newcommand{\Bc}{{\mathcal{B}}}
\newcommand{\ra}{{\rightarrow}}
\newcommand{\ift}{{\infty}}
\newcommand{\oc}{{\overline{c}}}

\newtheorem{lemma}{\textbf{Lemma}}
\newtheorem{theorem}{\textbf{Theorem}}
\newtheorem{remark}{\textbf{Remark}}
\newtheorem{assumption}{\textbf{Assumption}}
\newtheorem{corollary}{\textbf{Corollary}}

\DeclareMathOperator{\supp}{supp}
\DeclareMathAlphabet{\mathpzc}{OT1}{pzc}{m}{it}

\newcommand{\Sp}{{\mathpzc{S}}}

\newcommand{\Rp}{{\mathpzc{R}}}
\newcommand{\Cp}{{\mathpzc{C}}}
\newcommand{\Np}{{\mathpzc{N}}}
\newcommand{\Mp}{{\mathpzc{M}}}
\newcommand{\Op}{{\mathpzc{O}}}

\usepackage{pgfplots}
\pgfplotsset{compat=newest}
\usetikzlibrary{plotmarks}
\usepgfplotslibrary{patchplots}
\usepackage{grffile}
\usepackage{amsmath}

\title{\LARGE \bf
	Game Theoretical Approach to Sequential Hypothesis Test with Byzantine Sensors
}

\author{ Zishuo Li$^{1}$,  Yilin Mo$^{2}$, and Fei Hao$^{1}$
	\thanks{
		$^1$: Z. Li and F. Hao are with the School of Automation Science and Electrical Engineering, Beihang University, China. Email: {lizishuo1523@buaa.edu.cn; fhao@buaa.edu.cn}.
}
	\thanks{
		$^2$: Y. Mo is with the Department of Automation and BNRist, Tsinghua University, China. Email: {ylmo@tsinghua.edu.cn}.
}
}

\begin{document}
	
	\maketitle
	\thispagestyle{empty}
	\pagestyle{empty}

	\begin{abstract}
          In this paper, we consider the problem of sequential binary hypothesis test in adversary environment based on observations from $s$ sensors, with the caveat that a subset of $c$ sensors is compromised by an adversary, whose observations can be manipulated arbitrarily. We choose the asymptotic Average Sample Number (ASN) required to reach a certain level of error probability as the performance metric of the system. 
          The problem is cast as a game between the detector and the adversary, where the detector aims to optimize the system performance while the adversary tries to deteriorate it. 
          We propose a pair of flip attack strategy and voting hypothesis testing rule and prove that they form an equilibrium strategy pair for the game. We further investigate the performance of our proposed detection scheme with unknown number of compromised sensors and corroborate our result with simulation.
	\end{abstract}

	\section{INTRODUCTION}
	Recent advancements in communication technology and sensing elements have made networked sensor system more readily available in control systems, performing the function of observation, detection and monitoring.
	However, the reliance on communication and sparsely spacial distribution make the sensor system vulnerable in the presence of various cyber attacks such as measurement manipulation, communication block, false data injection, etc. Since malicious attacks, such as Stuxnet~\cite{STUXNET} and BlackEnergy malware~\cite{black_energy} may incur substantial damage on economy, ecosystem and even public safety, designing resilient networked system with secure detection, estimation and control algorithm has been recognized by both engineers and scholars as a significant research field.
	
	
	In this paper we consider the problem of detecting a binary state $\theta$ with $s$ sensors in adversarial environment. We assume $c$ out of $s$ sensors are compromised and their observations could be manipulated arbitrarily by the adversary. We introduce the Byzantine attack setting where system manager has no information about the exact set of corrupted sensors but only knows the cardinality of the set. The detection performance is evaluated by its Average Sample Number under prescribed level of significance (probability of error).
    We adopt a similar formulation as \cite{yanjiaqi2017CDC} where the problem is considered as a game between the detector and the attacker, in which the detector attempts to optimize the performance while the adversary intends to deteriorate it. A pair of strategy (attack strategy from the adversary and hypothesis testing scheme from the detector) is proposed and proved to be a Nash equilibrium pair for the game. Furthermore, scenario with unknown number of compromised sensors is investigated and choice of parameter in sequential test algorithm is discussed. 
	
	\textit{Related Work:}
	
	The study of sequential analysis (to the best of our knowledge) originated from 
	Abraham Wald et al.~\cite{wald_sequential_test_paper1945}\cite{wald_optimum1947} who proposed the Sequential Probability Ratio Test (SPRT) and proved its optimality in 1940s. Due to its wide applicability and optimality in hypothesis testing, sequential analysis has gained wide application in sensor network security design~\cite{recent_2_distributed}\cite{FAST_MOBILE_REPLICA}, change detection\cite{fellouris2018}\cite{change_detection+dos_attack}, signal anomaly detection~\cite{GNC}\cite{spoof}, etc.
	
	As threats to control systems from cyber attacks are increasing rapidly these days, studies about secure detection problem draw attention from researchers. The research efforts can be classified into two main directions: anomaly diagnose and resilient algorithm design. In the former one, anomaly diagnosis schemes are designed to reveal the existence of attack and trigger alarm and/or recovery mechanism.
	For example, the problem of revealing the existence of attacks or vulnerable part of the system that requires protection is considered in \cite{stealthy_attack_revealing} and \cite{hu_stateestimation_automatica}. 
	In the research about resilient algorithm design, researchers pursuit a design of secure system which has graceful performance degradation in the presence of attack. Since  
	attacks may not be eliminated immediately even if we know its existence because of the concealment of attackers in cyberspace, resilient algorithm design is preferred in the sense of safety guarantee. We choose resilient testing algorithm design as our research goal in this paper.
	
	The problem of resilient inference has been studied from various perspective recently including hypothesis testing \cite{yanjiaqi2017CDC}\cite{bayesian_detection}, change detection \cite{change_detection+dos_attack}\cite{fellouris2018}, 	state estimation~\cite{tabuada1}, etc. We focus on hypothesis testing problem.  
	Similar formulation of detecting a binary state with multiple sensors under Byzantine attack is studied by Ren et al.~\cite{xiaoqiang2018} recently and the problem of security-efficiency trade-off is raised. Moreover, the model is extended to multi-hypothesis testing and heterogeneous sensor scenario where game theoretic approach is adopted~\cite{xiao_multihypo_game} and sensor selection problem is investigated~\cite{xiao_different_sensor_deploy}.
	
	We consider the problem of detecting a binary state using sequential analysis in the sense that stopping time is determined by observations while some other researches use a prescribed number of observed samples, e.g. one-shot scheme~\cite{joao_adversarial}\cite{joao_integrity} and fixed time analysis \cite{xiaoqiang2018}.
	By making decision adapted to observations, Average Sample Number is saved (as can be seen in Remark \ref{rm:save_sample}) because sampling is stopped as soon as the existing observations possess enough preference on a certain hypothesis. 
	The efficiency of detection sampling in our paper is evaluated and optimized by integrating ASN into performance metric (see definition of delay in equation \ref{delay_def}). Similar methodology could be seen in the study of change detection (e.g. \cite{fellouris2018}\cite{veeravalli2012}).
	
	The rest of this paper is organized as follows: In Section~\ref{sec:problem}, we formulate the problem of binary hypothesis test and define the admissible attack and binary state detecting strategy as well as the performance metric. In Section~\ref{sec:pair}, we propose an attack strategy by flipping the distribution of observations from the compromised sensors and a resilient detection strategy by voting among all sensors. This pair of strategy is then proved to form an equilibrium pair for the game between attacker and detector. In Section~\ref{sec:extension}, the scenario where actual number of compromised sensors is unknown  is investigated and corresponding performance is quantified. Simulation result is provided in Section~\ref{sec:simu}, and Section VI finally concludes the paper. 
	
	
	\textit{Notations:} We denote by $\Zb^+$ the set of positive integers and by $\Rb$ the set of real numbers. We denote by $x \sim y$ when $x/y \ra 1 $. Cardinality of a finite set $\Sp$ is denoted as $|\Sp|$. Transpose of a vector or matrix is denoted by superscript $T$.
	
	\section{PROBLEM FORMULATION}
\label{sec:problem}	
	\subsection{Binary Hypothesis Testing}
	
	Suppose there is a binary state $\theta\in\{0,1\}$ detected by a group of $s$ sensors. At each discrete time index $k$, the observation from each sensor $i\in\Sp\triangleq \{1,2,...,s\}$ is collected by a fusion center. Let row vector $\bm{x}_i=[x_i(1),x_i(2),x_i(3),...]$ denote the sequence of observations from the $i$th sensor and column vector $\bm{x}(k)=[x_1(k),x_2(k),x_3(k),...,x_s(k)]^T$ denote the observations at time $k$ from all sensors. 
	We assume that all observations from different sensors at different time are independently identically distributed for each $\theta$. Simialr to notations in \cite{xiaoqiang2018}, when the hypothesis is false ($\theta=0$), probability measure generated by $x_i(k)$ is denoted as $\nu$ and it is denoted as $\mu$ when the hypothesis is true ($\theta = 1$).
	In other words, for any Borel-measurable set $\Bc\subseteq \mathbb R$, the probability that $x_i(k)\in \Bc$ equals to $\nu(\Bc)$ when $\theta = 0$ and equals to $\mu(\Bc)$ when $\theta = 1$.
	We denote the probability space generated by all measurements $\bm{x}(1),\,\bm{x}(2),\,\dots$ as $(\Omega,\,\mathcal F,\,\mathbb P_\theta)$
	, where for any $l\geq 1$
	\begin{align*}
	&\mathbb P_\theta(x_{i_1}(k_1)\in \Bc_1,\dots,  x_{i_l}(k_l)\in \Bc_l) \\
	&= \begin{cases}
	\nu(\Bc_1)\nu(\Bc_2)\dots\nu(\Bc_l)&\text{if }\theta = 0\\
	\mu(\Bc_1)\mu(\Bc_2)\dots\mu(\Bc_l)&\text{if }\theta = 1
	\end{cases},
	\end{align*}
	when $(i_j,k_j)\neq (i_{j'},k_{j'})$ for all $ j \neq j' $.
	The expectation taken with respect to $\mathbb P_\theta$ is denoted by $\mathbb E_\theta$.
	
	We further assume that probability measure $\nu$ and $\mu$ are absolutely continuous with respect to each other.
	Therefore, the log-likelihood ratio $L_i(k)$ of $x_i(k)$ is well-defined as
	\begin{equation}
	L_i(k) \triangleq \log\left(\frac{\rm d\mu}{\rm d\nu}(x_i(k))\right) ,
	\label{eq:loglikelihoodratio}
	\end{equation}
	where $\rm d\mu/ \rm d\nu$ is the Radon-Nikodym derivative.
	\subsection{Byzantine Attack}
	Let the (manipulated) observation received by the fusion center at time $k$ be
	\begin{equation}\label{rewrite_attack}
	\bm{x}'(k)=\bm{x}(k)+\bm{x}^a(k),
	\end{equation}
	where $\bm{x}^a(k) \in \mathbb{R}^s$ is the deflective vector injected by the attacker at time $k$. By adding values to the real observations $\bm{x}(k)$, the attacker can rewrite them to arbitrary value they assign. We have the following assumptions on the attacker.
	
	\begin{assumption}[Sparse Attack] \label{assumpt:sparse}
		There exists an index set $\Cp\subseteq \Sp $ with $|\Cp| = c$ such that $\bigcup_{k=1}^{\infty} \supp\left\{\bm{x}^a(k)\right\} =\Cp$ where
		$\supp(\bm{x})\triangleq \left\{i\in\Sp:{x}_i\neq0 \right\}$ is the support of vector $\bm{x}$.
		Furthermore, the system knows the cardinality $c$, but it does not know the set $\Cp$.
	\end{assumption}
\begin{remark}
	It is conventional in the literature (e.g. \cite{fellouris2018} \cite{soltan2013_misbehavenodes}\cite{HAN2015_Bayesian_Detection_Byzantine_Data}) to assume that the attacker possesses limited resources, i.e., the number (or percentage) of compromised sensors is fixed and is known by the system manager. The value of $c$ can also be seen as a design parameter representing the tolerance of sensor corruptions in the system.
\end{remark}

We denote by $\Np \triangleq \Sp\setminus \Cp$ the honest (not affected by attack) sensor. The information the attacker have access to is assumed as follows:
	\begin{assumption}[Attacker Knowledge] \label{assumpt:att_know}
		(1) The attacker knows the probability measure, i.e. $\mu$ and $\nu$;
		(2) The attacker knows the real system state $\theta$; 
		(3) The attacker knows the real observation from all compromised sensors from the beginning to the present time instant.
	\end{assumption}

	\begin{remark}
	The only restriction on the attack strategy is that the set of compromised sensors is fixed (from Assumption \ref{assumpt:sparse}). The attacker have adequate knowledge about the system and can carry out complex attack strategies such as time-varying or probabilistic ones. This assumption is conventional in literature concerning the worst-case attacks (e.g. \cite{marano_byzantine_attack}). Nevertheless, assuming the adversary to be powerful when designing system would make sure its security and is in accordance with the Kerckhoffs's principle.
	\end{remark}

	An admissible attack strategy is a mapping from attacker's information set to the bias vector that satisfies Assumption \ref{assumpt:sparse}. Let the compromised sensor index set $\Cp=\{i_1,i_2,\cdots,i_c\}$. 
	Define $\bm{X}_\Cp(k)$ as the matrix formed by true measurements from time 1 to k at compromised sensors:
	\begin{equation*}
	\bm{X}_\Cp(k)\triangleq[\bm{x}_\Cp(1),\bm{x}_\Cp(2),\cdots,\bm{x}_\Cp(k)]\in \Rb^{c\times k}
	\end{equation*}
	with
	\begin{equation*}
	\bm{x}_\Cp(k)\triangleq[x_{i_1}(k),x_{i_2}(k),\cdots,x_{i_c}(k)]^T\in \Rb^{c\times 1}.
	\end{equation*}
	
	Similar to $\bm{X}_\Cp(k)$, $\bm{X}^a(k)\in \Rb^{s\times k}$ is defined as the matrix formed by bias vectors $\bm{x}^a(k) \in \Rb^{s \times 1}$ from time 1 to $k$. The injected bias vector is designed by the attacker based on its information set, i.e.
	\begin{equation}\label{eq:attack_strategy}
	\bm{x}^a(k)=g\left(\bm{X}_\Cp(k),\bm{X}^a(k-1),\theta,k\right),
	\end{equation}
	where $g$ is a measurable function of accessible observations $\bm{X}_\Cp(k)$, history attacks $\bm{X}^a(k-1)$, real state $\theta$ and time $k$ such that $\bm{x}^a(k)$ satisfies Assumption 1.
	Denote the probability space generated by all manipulated observations $\bm{x}'(1),\bm{x}'(2),\dots$ as $(\Omega, \mathcal{F}, \Pb^g_\theta)$ where $\theta$ is the real state. The corresponding expectation is denoted as $\Eb^g_\theta$.
	\subsection{Performance Metric}
	The detector at time $k$ is defined as a mapping from the manipulated observation matrix to the set of decision:
	$$f_k:\bm{X}'(k)\ra \{continue,0,1\},$$
	where $continue$ denote taking next observation at time $k+1$ because existing knowledge is not enough to make a decision. Decision 0 and 1 denote stop taking observations and choose hypothesis $H_0$ and $H_1$ respectively.
	System's strategy $f \triangleq (f_1,\,f_2,\,\cdots)$ is defined as an infinite sequence of detectors from time $1$ to $\infty$.
	
	Based on the definition of detection strategies, the stopping time $T$ representing the time that the test terminates is a $\{\Fc'_t\}$-stopping time, where $\Fc'_t$ is a $\sigma$-field of all the (manipulated) observations from time 1 to $k$: $\Fc'_t=\sigma\{\bm{X}'(k)\}.$
	Define the worst case Average Sample Number (detection delay) under attack $g$ as 
	\begin{equation}\label{delay_def}
		D(T)\triangleq\max_{\theta=0,1}\Eb^g_\theta[T].
	\end{equation}

	Denote the probability of Type-I and Type-II error\footnote{In statistical hypothesis testing, a type-I error is rejection of a true null hypothesis $H_0$, while a type-II error is the failure to reject a false null hypothesis.}
	of the binary hypothesis testing task as $\alpha$ and $\beta$ respectively, e.g. $\alpha\triangleq \Pb_0[f_T=1], \beta\triangleq \Pb_1[f_T=0]$. As a detector needs to make decisions based on as few samples as possible under error probability constraints which vary in different situations, we consider the asymptotic performance as error probability tends to zero:
	\begin{equation}\label{eq:per_def}
		\gamma(f,g)\triangleq \lim_{\alpha=\beta\ra0^+}	\frac{\log(1/\alpha)}{D(T)}.
	\end{equation}
\begin{remark}
	By definition, $\gamma(f,g)\geq 0$ for any admissible $f$ and $g$ because $\alpha\leq 1$ and $D(T)> 0$.
	The performance integrates error probabilities $\alpha,\beta$ with detection delay $D(T)$ which we hope to be small at the same time. It means larger $\gamma$ indicates better detection performance.
\end{remark}
	\begin{remark}
	The performance $\gamma$ is determined by both the detection rule $f$ and attack strategy $g$ so it is denoted as $\gamma(f,g)$. The system manager intends to design a resilient detector $f$ to maximize $\gamma$ while the attacker needs malicious attack $g$ to minimize $\gamma$. 
\end{remark}

	In this paper, we intend to propose a pair of strategy $(f^*,g^*)$, such that for any strategies $f$ and $g$, the following inequality holds:
	\begin{equation} \label{eq:equ_pair}
	\gamma(f,g^*) \leq \gamma(f^*,g^*) \leq \gamma(f^*,g).
	\end{equation}
	 
	As a result, the pair of strategy $(f^*,g^*)$ reaches a Nash equilibrium (which is not necessarily unique). In other words, given strategy of one player as $f^*$ or $g^*$, the other player do not have a strictly better strategy. We present the strategy pair in the next section.

	\section{Equilibrium Strategy Pair}	\label{sec:pair}
	In this section we present an attack strategy and a detection scheme and prove that they can form a Nash equilibrium pair.
	\subsection{Preliminaries Results}
	Before we go on, we first present some basic results of hypothesis testing scheme without attack which will be helpful for future discussion.
	Denote the Kullback-Leibler (K--L) divergences between those two distribution we are trying to distinguish (i.e. $\mu$ and $\nu$) as
	\begin{align*}
	I_1\triangleq \int_{x\in\Rb}\log\left[\frac{\rm{d} \mu(x)}{\rm{d}\nu(x)}\right] {\rm d}\mu(x), 
	I_0\triangleq-\int_{x\in\Rb}\log\left[\frac{\rm{d}\mu(x)}{\rm{d}\nu(x)}\right] {\rm d}\nu(x)
	\end{align*}
	
	To avoid degenerate problems, we adopt the following assumptions:
	\begin{assumption}
		The K--L divergences are well-defined, i.e., $0<I_0,I_1<\infty$.
		\label{asm:finiteKLs}
	\end{assumption}

	We introduce a sequential test strategy for multiple sensor based on Sequential Probability Ratio Test proposed by Wald \cite{wald1947}.
	We denote the cumulative log-likelihood ratio of sensor $i$ at time $n$ by $S_i(n)$ and the one summing over set $\Mp$ by $S_\Mp(n)$:
	\begin{equation}\label{eq:def_S_M}
	S_i(n)\triangleq \sum_{k=1}^{n}  L_i(k)	,\quad S_\Mp(n)\triangleq \sum_{i\in\Mp} S_i(n),
	\end{equation}
	where $\Mp\subseteq\Sp$. The decision is taken according to whether the prescribed threshold is crossed, i.e.
	
	\begin{equation}\label{eq:simple_SPRT_1}
	f_k=
	\left\{
	\begin{array}{cc}
	
	0 ,& S_\Mp(k)\leq -a\\
	continue ,& -a<S_\Mp(k)<b\\
	1 ,& S_\Mp(k)\geq b
	\end{array},
	\right.
	\end{equation}
	where $a,b>0$ are chosen to regulate error probabilities $\alpha,\beta$. Denote the defined detection rule based on summed log-likelihood ratio from sensors in $\Mp$ as $f_\Mp$. We have the following lemma quantifying performance of this test (called sum-SPRT) in the absence of attack. The proof is provided in Appendix (Section~\ref{ap:lemma_M}).
	
	\begin{lemma}\label{lm:M_sum-SPRT_per+opt}
		Define $I\triangleq \min\{I_0,I_1\}$, for all admissible test rule $f$ based on sensor information in $\Mp$ ,
		\begin{equation}\label{eq:M_sum}
			\gamma(f,g=\bm{0})\leq\gamma(f_\Mp,g=\bm{0})= |\Mp|\cdot I,
		\end{equation}
		where $g=\bm{0}$ means the attacker is absent.	
	\end{lemma}

	\begin{remark}
	The performance of $f_\Mp$ is proportional to the number of sensors $|\Mp|$ and the constant $I$ defined by K-L divergence. Constant $I$ who represents the "distance" of two distributions could be treated as a basic unit of performance.
	\end{remark}
%
	
	Now we move on to consider the detection problem under attack. We assume $s>2c$ to prevent trivial problems in the rest of paper if without further notice.
	\subsection{Attack Strategy}
	In this subsection we show an attack strategy where the attacker flips the distribution of the compromised sensor observations under different states to confuse the detector. We denote it as $g^*$ (named flip attack) and it is defined in the following:
	
	Denote sensor index set of the first $c$ sensors as $\Op_1\triangleq\{1,2,\dots,c\}$ and the set of last $c$ sensors as $\Op_2 \triangleq \{s - c + 1, s - c + 2,\dots, s\}$.
    If $\theta=0$, generates random observations $\tilde{x}_i(k)$ at time $k$ for every sensor $i\in \Op_1$ according to the opposite distribution $\mu$, i.e. for each Borel set $\Bc$
    \begin{equation}\label{eq:def_att_theta0_1}
    \Pb[\tilde{x}_i(k)\in\Bc]=\mu(\Bc), \ \theta=0,i\in \Op_1.
    \end{equation}
    Then design bias data to make sure the final observations $x_i(k)+x^a_i(k)$ of sensors in $\Op_1$ is the same as $\tilde{x}_i(k)$:
    
    \begin{equation}\label{eq:def_att_theta0_2}
    	x^a_i(k)=\tilde{x}_i(k)-{x}_i(k), \  \theta=0,i\in \Op_1.
    \end{equation}
    
    If $\theta=1$, observations in $\Op_2$ is manipulated in similar way.
    \begin{equation}\label{eq:def_att_theta1_1}
    \Pb[\tilde{x}_i(k)\in\Bc]=\nu(\Bc), \ \theta=1,i\in \Op_2.
    \end{equation}
    \begin{equation}\label{eq:def_att_theta1_2}
    x^a_i(k)=\tilde{x}_i(k)-{x}_i(k), \  \theta=1,i\in \Op_2.
    \end{equation}
    
    For sensors not mentioned above, the bias value $x^a_i(k)=0$. By this operation, the following inequality of performance holds.
    
%
%
	 
	\begin{theorem}\label{th:attack}
		For any admissible detection strategy $f$ we have
		\begin{equation}\label{eq:attack_upper}
		\gamma(f,g^*)\leq (s-2c) I.
		\end{equation}
	\end{theorem}
	
	\begin{remark}
		The coefficient $(s-2c)$ indicates that the detector will have positive performance when less than half of the sensors are compromised.
		It also implies every increase of compromised sensor will incur two units of performance decrease. The result follows Theorem 3(2) in \cite{xiaoqiang2018}.
	\end{remark}
	\begin{proof}
	Under attack $g^*$, for either $\theta=0$ or $\theta=1$, sensors in $\Op_1$ will follow distribution $\mu$ and sensors in $\Op_2$ will follow distribution $\nu$. In other words, only sensors in $\Sp\setminus(\Op_1\cup\Op_2)$ have different distributions under different $\theta$.
	Since we assume $s>2c$, $\Sp\setminus(\Op_1\cup\Op_2)\neq\emptyset$. If we define $\Mp=\Sp\setminus(\Op_1\cup\Op_2)$, according to Lemma \ref{lm:M_sum-SPRT_per+opt}, 
	$$\gamma(f,g^*)\leq \gamma (f_{\Mp},g=\bm{0})= |\Sp\setminus(\Op_1\cup\Op_2)|I=(s-2c)I.$$
	Thus, equation (\ref{eq:attack_upper}) is obtained.	
	\end{proof}

	\subsection{Detection Strategy}
	In this section we present a detection strategy that could form a Nash equilibrium pair with flip attack $g^*$.
	Before we present the detection rule, we first define some notations.
	
	First we define the stopping time of single threshold test for each sensor $i$ in (\ref{eq:def_taub})(\ref{eq:def_taua}). Similar to basic SPRT, those two thresholds are denoted as $-a<0<b$:
	\begin{align}
	\tau^+_{i}(b)\triangleq &\inf_{k\in\Zb^+}\{S_{i}(k)\geq b\} \label{eq:def_taub}. \\
	\tau^-_{i}(a)\triangleq &\inf_{k\in\Zb^+}\{S_{i}(k)\leq -a\}  \label{eq:def_taua}.
	\end{align}
	Then sort those stopping time of the same threshold in an ascending order and denote them as $\tau^-_{(i)}(a),\tau^+_{(i)}(b)$ : 
	$$\tau^-_{(1)} (a)\leq \tau^-_{(2)} (a)\leq\cdots \leq\tau^-_{(s)} (a),$$
	$$\tau^+_{(1)} (b)\leq \tau^+_{(2)} (b)\leq\cdots \leq\tau^+_{(s)} (b).$$	
	Define $r$ as the parameter of decision rules with $s/2<r\leq s$ and the voting rule $f^{(r)}$ is defined as taking corresponding hypothesis the first time when there have been $r$ crossing of the same threshold. The rule is showed formally in the following. For each time $k$,
	\begin{equation}\label{eq:def_vote_rule}
	f^{(r)}_k=
	\left\{
	\begin{array}{cc}
	continue ,& k<min\{\tau^-_{(r)}(a),\tau^+_{(r)}(b)\} \\
	0 ,& k=\tau^-_{(r)}(a)<\tau^+_{(r)}(b)\\
	1 ,& k=\tau^+_{(r)}(b)<\tau^-_{(r)}(a)\\
	0\ or\ 1, & k=\tau^+_{(r)}(b)=\tau^-_{(r)}(a)\\
	\end{array}.
	\right.
	\end{equation}
	The decision $0\ or\ 1$ means stop sampling and take $H_0$ or $H_1$ with the same probability 0.5. Denote the detection strategy defined above as $f^{(r)}\triangleq\{f^{(r)}_1,f^{(r)}_2,\dots\}$.	We denote the stopping time of detection rule $f^{(r)}$ as $T^{(r)}$.
	
	Before we show the performance of detection strategy, we provide some preliminary results of stopping times and error probabilities in absence of attack whose proof is provided in Appendix (Section~\ref{ap:theo_prelimi})
	\begin{theorem}\label{th:basic_char}
			\begin{align}
			(1)&\lim_{a=b\ra\ift}\Eb_0\left|\frac{\tau^-_{(r)}(a)}{a}-\frac{1}{I_0}\right|=0,  \label{eq:tau_r}
			\lim_{a=b\ra\ift}\Eb_1\left|\frac{\tau^+_{(r)}(b)}{b}-\frac{1}{I_1}\right|=0  \\	
			(2)&\lim_{a=b\ra\ift}\frac{\Eb_0[{T}^{(r)}]}{a}\leq\frac{1}{I_0},\ \lim_{a=b\ra\ift}\frac{\Eb_1[{T}^{(r)}]}{b}\leq\frac{1}{I_1} \label{eq:T_r}\\
			(3)&\lim_{a=b\ra\ift} \frac{1}{b}\log \Pb_0[\tau^+_{(r)}(b)\leq\tau^-_{(r)}(a)]\leq -r \label{eq:alpha} \\
			&\lim_{a=b\ra\ift} \frac{1}{a}\log \Pb_1[\tau^-_{(r)}(a)\leq\tau^+_{(r)}(b)]\leq -r \label{eq:beta} 
			\end{align}
	\end{theorem} 
	
	

		Based on Theorem~\ref{th:basic_char}  we are ready to show the performance of our detection rule with carefully designed $r$.
	\begin{theorem}\label{th:detect}
		For any admissible attack strategy $g$, fix $r=s-c$ and denote $f^*\triangleq f^{(s-c)}$. We have
		$$\gamma(f^*,g)\geq (s-2c)I.$$
	\end{theorem}


	\begin{proof}
		We show the following inequalities for arbitrary attack $g$ (notice that $\Pb_\theta^g$ and $\Eb_\theta^g$ denote probability and expectation under attack $g$)
		\begin{align}
		\Eb_1^g[\tau^+_{(r)}(b)]&\leq \Eb_1[\tau^+_{(r+c)}(b)]. \label{eq:CT_un_H1}\\
		\Eb_0^g[\tau^-_{(r)}(a)]&\leq \Eb_0[\tau^-_{(r+c)}(a)]. \label{eq:CT_un_H0}\\
		\Pb^g_1[\tau^-_{(r)}(a)\leq \tau^+_{(r)}(b)]&\leq \Pb_1[\tau^-_{(r-c)}(a)\leq \tau^+_{(r-c)}(b)]. \label{eq:FA_un_H1} \\
		\Pb^g_0[\tau^-_{(r)}(a)\geq \tau^+_{(r)}(b)]&\leq \Pb_0[\tau^-_{(r-c)}(a)\geq \tau^+_{(r-c)}(b)]. \label{eq:FA_un_H0}
		\end{align}
		
		Considering the readability of this section, the proof of these four inequalities above is shown in Appendix (Section~\ref{ap:four_ineq}). The left hand side of them are expected stopping times and error probabilities under attack while the right hand side are the ones without attacks. They are obtained considering cumulative log-likelihood ratio and threshold-reached time in the worst case given all admissible attack.
		
		We are ready to quantify the performance under attack with the help of inequalities (\ref{eq:CT_un_H1}) to (\ref{eq:FA_un_H0}).
		On one hand, detection delay can be upper bounded based on (\ref{eq:tau_r}) and (\ref{eq:CT_un_H1}):
		$$\Eb^g_1[T^{(r)}]\leq \Eb^g_1[\tau^+_{(r)}(b)]\leq \Eb_1[\tau^+_{(r+c)}(b)]\sim 
		\frac{b}{I_1}.$$
		in which the first inequality comes from definition of voting rule (\ref{eq:def_vote_rule}).
		On the other hand, error probability can be quantified based on (\ref{eq:beta}) and (\ref{eq:FA_un_H1}):
		\begin{align*}
		\beta\leq&\Pb^g_1[\tau^-_{(r)}(a)\leq \tau^+_{(r)}(b)]\\
		\leq &\Pb_1[\tau^-_{(r-c)}(a)\leq \tau^+_{(r-c)}(b)] \leq Ce^{-(r-c)a},
		\end{align*}
		where C is a constant term.
		Those two inequalities imply
		\begin{align*}
		\lim_{\alpha=\beta\ra 0^+} \frac{\log(1/\beta)}{\Eb^g_1[T^{(r)}]}
		&=\lim_{a=b\ra \ift} \frac{\log(1/\beta)}{\Eb^g_1[T^{(r)}]} \\
		&\geq \lim_{a=b\ra \ift} \frac{(r-c)\cdot a}{b/I_1} = (r-c)I_1.
		\end{align*}
		When $\theta=0$, similar results could be derived from equation (\ref{eq:CT_un_H0}) and (\ref{eq:FA_un_H0}). Thus, by replacing $r$ with $s-c$, the final result is obtained:
		$$ \gamma(f^*,g)\geq (s-c-c)\min\{I_0,I_1\} =(s-2c)I.$$
		The proof is completed.
	\end{proof}
	
	Combining Theorem \ref{th:attack} and \ref{th:detect}, we are ready to show the Nash equilibrium pair of strategies.
	\begin{theorem}\label{th:equilibrium}
		 Detection strategy $f^*$ defined in (\ref{eq:def_vote_rule}) with $r=s-c$ and attack strategy $g^*$ defined in (\ref{eq:def_att_theta0_1}) to (\ref{eq:def_att_theta1_2}) form a Nash equilibrium, i.e. for any admissible detection rule $f$ and attack $g$,
		$$\gamma(f,g^*) \leq \gamma(f^*,g^*)=(s-2c)I \leq \gamma(f^*,g).$$ 
	\end{theorem}
	\begin{proof}
		Set the detector in Theorem \ref{th:attack} as $f^*$ and attack in Theorem \ref{th:detect} as $g^*$ and we can obtain $\gamma(f^*,g^*) \geq(s-2c)I$ and $\gamma(f^*,g^*)\leq(s-2c)I$ at the same time. Substituting $(s-2c)I$ with $\gamma(f^*,g^*)$ in theorem \ref{th:attack} and \ref{th:detect} leads to the result.
	\end{proof}
	\begin{remark}
		The payoffs for players of the game are $\gamma(f,g)$ (for detector $f$) and $-\gamma(f,g)$ (for attacker $g$). Notice that the strategy set for this game is non-compact, the Nash equilibrium does not necessarily exist. Our result actually proved the existence of Nash equilibrium in addition to a pair of specific strategy.
	\end{remark}
	\begin{remark}\label{rm:save_sample}
		Since the definition of $\gamma(f,g)$ can also be used to evaluate non-sequential detection schemes, we are able to compare their performance with ours. We define 
		$$\tilde{I}\triangleq -\log\left[\inf_{w\in\Rb}\left\{\int_{x\in\Rb}
		\left(\frac{{\rm d}\mu(x)}{{\rm d}\nu(x)}\right)^w {\rm d} \nu(x) \right\}\right].$$
		It has been shown in~\cite{xiaoqiang2018} Theorem 2 that $0<\tilde{I}< I$. The detector performance defined in~\cite{xiaoqiang2018} is the same as ours for fixed sample detecting scheme. However, the value of detector performance in that paper is $(s-2c)\tilde{I}$ which is smaller than ours. In this sense, our scheme is more sample-efficient because the sampling is terminated as soon as there is enough statistical information indicating the real hypothesis.
	\end{remark}
	\begin{remark}\label{rm:computation_better}
		Single time step computation complexity of our detection scheme is $O(s)$ as computing $S_i(k)$ and voting among sensors both have a complexity of $O(s)$. Therefore, the computational complexity is lower than the result in~\cite{xiaoqiang2018} where the sorting algorithm cause a computational complexity of $O(s\log s)$. Moreover, voting detection algorithm is more easily applied to distributed computing because the sensors do not need to send the actual observations to the control center but only need to inform whether the threshold is crossed. System based on our detection algorithm have less information transmission pressure and is more likely to achieve better efficiency and resilience.
	\end{remark}
	\section{Extensions}\label{sec:extension}
	In the previous section, we assume the number of compromised sensors $c$ is known to the system manager. However, in practice the real value may be unknown and what we have is a estimation of its upper bound. It can be seen as a design parameter denoting how many sensor corruptions the system can tolerate. 
	In this section, we study the condition where we have an upper bound $\oc$ and the actual number of compromised sensors $c$  can take value in $\{0,1,2,\dots,\oc\}$.
	
	We denote the voting detection rule with $r=s-\oc$ as $\tilde{f}\triangleq f^{(s-\oc)}$. We have the following Theorem revealing the lower bound of its performance.
	\begin{theorem}\label{th:extension1}
		Given detector $\tilde{f}$, assume $c$ is the actual number of compromised sensors  and $c\leq\oc<s/2$. Under any admissible attack, we have 
		$$\gamma(\tilde{f},g)\geq (s-\oc-c)I.$$
	\end{theorem}
	\begin{proof}
	In this setting, Theorem \ref{th:detect} still holds true and the only difference is the choice of $r$. Thus, the result is obtained by substituting $s-c$ with $s-\oc$.
	\end{proof}
	\begin{remark}
		The performance loss is in proportional to the sum of estimation number of corruption $\oc$ and the actual number of corruption $c$. If $c$ is fixed, excessive $\oc>c$ will cause unnecessary performance loss.
	\end{remark}

The result in Theorem~\ref{th:extension1} implies the performance lower bound is $(s-\oc)I$ when all sensors are benign. We present it in the following Corollary formally.
\begin{corollary}\label{rm:effi_lower_bound}
	When there is no attack, i.e. $c=0$, performance is lower bounded: $$\gamma(\tilde{f},g=\bm{0})\geq (s-\oc)I.$$
\end{corollary}
\begin{remark}
	$\gamma(\tilde{f},g=\bm{0})$ could be seen as the detection efficiency of voting rule at normal operation (attacker is absent). The increasing of $\oc$ will sacrifice detection performance in absence of attack while gaining better system resilience. Thus, sufficient knowledge about the attacker (e.g. how many sensors will be compromised) will be helpful for system efficiency-security trade off. Since the equilibrium strategy pair is not unique, questing for a detection rule who can achieve maximum performance when the attack is present and absent simultaneously is meaningful and could be our future work.
\end{remark}
\section{Simulation}\label{sec:simu}
In this section, we provide some numerical examples to verify the results established in the previous sections. We assume the observations of sensors follow i.i.d. distribution of $N(-1,1)$\footnote{$N(p,q^2)$ represent Normal distribution with mean $p$ and variance $q^2$.} when $\theta=0$ and $N(1,1)$ when $\theta=1$. In this case $I=I_0=I_1=2$. 

We set $s=10$ and $c$ varies from 0 to 4. In Fig. \ref{Fig:euqlibrium}, detection and attack strategy are $f^*$ and $g^*$ respectively. We calculate detection delay $D(T)$ and error probability $\alpha$ with threshold $a=b$ vary from $5\times10^0$ to $1\times10^5$ for each fixed $c$. The result $\frac{\log(1/\alpha)}{D(T)}$ is normalized by $I$ and should tend to $s-2c$ according to Theorem \ref{th:equilibrium}. To simulate the error probability with higher accuracy, we adopt the importance sampling approach \cite{rubinstein2016simulation}.
\begin{center}
	\begin{figure}[ht]
%
%
\definecolor{mycolor1}{rgb}{0.00000,1.00000,1.00000}%
\begin{tikzpicture}

\begin{axis}[%
width=2.5in,
height=1.8in,
at={(0.758in,0.481in)},
scale only axis,
xmode=log,
xmin=1,
xmax=100000,
xminorticks=false,
ymin=0,
ymax=10,
xlabel=threshold ${a,b}$ ${(a=b)}$ ,
ylabel=$\gamma{(f^*,g^*)}/ I$, 
axis background/.style={fill=white},
xmajorgrids,
ymajorgrids
]
\addplot [color=blue, line width=1.0pt, forget plot]
  table[row sep=crcr]{%
5	0.177009103527647\\
10	0.597202373440582\\
50	1.27558006193923\\
100	1.48714766759008\\
500	1.63589176494012\\
1000	1.79527886291003\\
5000	1.87871255158055\\
10000	1.90627289864839\\
50000	1.95984827231863\\
100000	1.97508230145343\\
};
\addplot [color=red, dashed, line width=1.0pt, forget plot]
  table[row sep=crcr]{%
5	2.11356211292114\\
10	2.40482075925231\\
50	3.34264255102409\\
100	3.43856186200425\\
500	3.72161401558856\\
1000	3.73988211015632\\
5000	3.89554006965826\\
10000	3.90368404316336\\
50000	3.96820670218695\\
100000	3.98048046250919\\
};
\addplot [color=green, dotted, line width=1.0pt, forget plot]
  table[row sep=crcr]{%
5	4.12118695708407\\
10	4.50182344900988\\
50	5.15249723845063\\
100	5.29059390861958\\
500	5.72738989237863\\
1000	5.77332388149342\\
5000	5.9007402258757\\
10000	5.94103918993742\\
50000	5.9618442509706\\
100000	5.97665927078343\\
};
\addplot [color=mycolor1, dashdotted, line width=1.0pt, forget plot]
  table[row sep=crcr]{%
5	5.52120457440255\\
10	6.40712336907693\\
50	7.08376588159133\\
100	7.33452763069052\\
500	7.56969931360208\\
1000	7.70566359148392\\
5000	7.88828607607036\\
10000	7.91719441468106\\
50000	7.96336871236542\\
100000	7.97758968188271\\
};
\addplot [color=black,  line width=1.0pt, forget plot]
  table[row sep=crcr]{%
5	7.46654891208226\\
10	7.26551741141603\\
50	8.79342187614214\\
100	9.03767581866198\\
500	9.52452750817051\\
1000	9.68949470989284\\
5000	9.86935913752838\\
10000	9.91043573842608\\
50000	9.96231472108835\\
100000	9.96702424777133\\
};
\end{axis}
\end{tikzpicture}%
	\caption{Normalized performance of equilibrium strategy pair $(f^*,g^*)$ when $s = 10$ for $c = 0$ (black solid line), $c = 1$ (cyan dash dot line), $c = 2$ (green dot line) , $c = 3$ (red dash line) and $c = 4$ (blue solid line).}
	\label{Fig:euqlibrium}
	\end{figure}
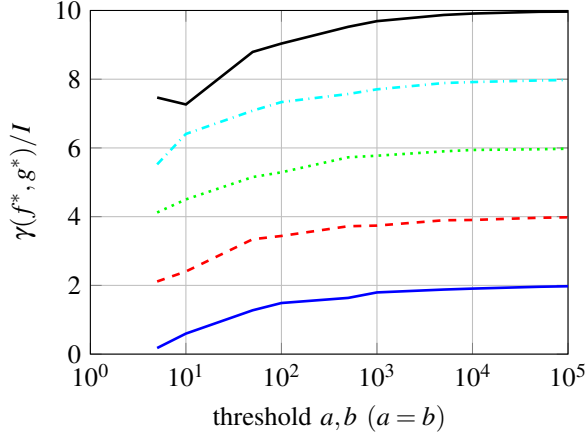
\end{center}

\section{Conclusion}\label{sec:conclusion}
In this paper, we formulate the problem of binary sequential detection in adversarial environment as a game between the detector and the attacker. Detection performance is defined asymptotically by both error probability and Average Sample Number as error probability tends to zero and this value is integrated in the game as payoff which the detector intends to maximize while the attacker intends to minimize. We propose a pair of detection rule and attack strategy and prove them to be an equilibrium pair of the game.
Furthermore, the performance in condition where number of compromised sensor is unknown and where all sensors are benign is quantified. The choice of detection rule parameter $r$ is discussed and result is corroborated by numerical simulation. The future work includes the trade-off between system's security and efficiency as well as discussion about achievability of optimal security and efficiency.
	

	\section{Appendix}
	\subsection{Proof of Lemma \ref{lm:M_sum-SPRT_per+opt}}\label{ap:lemma_M}	
	First we claim the inequality in (\ref{eq:M_sum}) is true by contradiction to optimum character of SPRT in \cite{wald_optimum1947}.
	
	Denote $|\Mp|I$ as $I_\Mp$. Assume there exists a detection rule $f'$ that satisfy $\gamma(f',g=\bm{0})=I'>I_\Mp$. We consider the condition where probability of type-I error equals to that of type-II error. Those error probabilities is denoted as $p_\Mp\triangleq\alpha(f_\Mp)=\beta(f_\Mp)$ and $p'\triangleq\alpha(f')=\beta(f')$. The stopping time is denoted as $T_\Mp$ and $T'$ respectively. By definition $\forall \epsilon>0$ there exists $P>0$ such that $\forall\ 0<p_\Mp,p'<P$
	$$\left|\frac{\log(1/p_\Mp)}{D(T_\Mp)}-I_\Mp\right|<\epsilon,$$
	$$\left|\frac{\log(1/p')}{D(T')}-I'\right|<\epsilon.$$
	Let $p'=p_\Mp$ and choose $\epsilon<(I'-I_\Mp)/2$ then one obtains
	\begin{align*}
	D(T')<\frac{\log(1/p')}{I'-\epsilon}<\frac{\log(1/p_\Mp)}{I_\Mp+\epsilon}<D(T_\Mp),
	\end{align*}
	which implies $\min_\theta \{\Eb_\theta(T')\}$ $<$ $\min_\theta\{\Eb_\theta(T_\Mp)\}$. That contradicts to $\{\Eb_\theta(T')\geq \Eb_\theta(T_\Mp)\},\theta=0,1$ (coming from optimality of Wald's Test). Thus completes the proof.
	\begin{remark}\label{rm:multi_sensor_opt_general}
		The optimality of SPRT is originally proved for single sensor scenario. As the observation from sensors are i.i.d. distributed, they could be formulated as a $s$ dimension observation vector and each hypothesis represents a joint distribution. For example, when $\theta=0$, $\nu(X_1=x_1,\dots,X_s=x_s)\triangleq \prod_{i=1}^{s}\nu(X_i=x_i)$. And the optimality (minimal Average Sample Number with same error probability) still holds.
	\end{remark}
	 
	In the following we show the asymptotic performance of sum-SPRT. Some of useful asymptotic properties of $\alpha,\beta$ and $\Eb_\theta[T_\Mp]$ has been provided by Berk \cite{berk_asy_aspects_of_SA1973} and we show some equivalent statements in the following Lemma.
	\begin{lemma}\label{lm:asy_property_Berk}
		Assume $a,b>0$, $\Mp\neq\emptyset$ and we have the following results hold with probability one: \\
		$$\lim_{\alpha=\beta\ra0^+}\frac{1}{a}\log\frac{1}{\beta}=1,\, \
		\lim_{\alpha=\beta\ra0^+}\frac{1}{b}\log\frac{1}{\alpha}=1,$$
		$$\lim_{\alpha=\beta\ra0^+}\frac{\Eb_0[T_\Mp]}{a}=\frac{1}{|\Mp|I_0},\,  \
		\lim_{\alpha=\beta\ra0^+}\frac{\Eb_1[T_\Mp]}{b}=\frac{1}{|\Mp|I_1}.$$
	\end{lemma}
	\begin{proof}
	Those results of single sensor (i.e. $|\Mp|=1$) follow straightforwardly from \cite{berk_asy_aspects_of_SA1973} Theorem $2.1$ and $2.2$. We focus on the generalization of multi-sensor scenario. 
	
	Firstly, the asymptotic characteristic of error probabilities do not rely on sensor numbers for the same reason as in Remark \ref{rm:multi_sensor_opt_general} and therefore first two equalities are true. We concentrate on the Average Sample Number. First we have the following from \cite{berk_asy_aspects_of_SA1973}. For every $i\in\Sp$
	$$\lim_{\alpha=\beta\ra0^+}\frac{\Eb_0[T_i]}{a}=\frac{1}{\Eb_0[L_i(k)]},\,  \
	\lim_{\alpha=\beta\ra0^+}\frac{\Eb_1[T_i]}{b}=\frac{1}{\Eb_1[L_i(k)]},$$
	where $T_i$ is the stopping time of $T_\Mp$ when $\Mp$ is a singleton set $\{i\}$. According to definition in (\ref{eq:def_S_M}), the summed log-likelihood ratio over set $\Mp$ at time $k$ is $\sum_{i\in\Mp}L_i(k)$. Therefore,
	$$\Eb_0\left[\sum_{i\in\Mp}L_i(k)\right]=\sum_{i\in\Mp}\Eb_0\left[L_i(k) \right]=|\Mp|I_0.$$
	The similar result could be obtained when $\theta=1$. Proof of Lemma \ref{lm:asy_property_Berk} is accomplished.
	\end{proof}

	We are ready to verify the equality in (\ref{eq:M_sum}).
	For readability purposes, the limitations without subscripts in the following equation means limits as $\alpha=\beta$ tends to $0^+$. With results above, noticing $a\sim b$, one obtains
	\begin{align*}
		\gamma(f_\Mp,g=\bm{0})=&\frac{\lim \log(1/\alpha)}{\max\{\lim\Eb_0[T_\Mp],\lim\Eb_1[T_\Mp]\}} \\
		=&\min \left\{ \lim\frac{b}{a} \cdot \lim\frac{a}{\Eb_0[T_\Mp]}, \lim\frac{ b}{\Eb_1[T_\Mp]} \right\}\\
		=&\min\left\{|\Mp|I_0,|\Mp|I_1\right\}=|\Mp|\cdot I.
	\end{align*}
	Proof of Lemma \ref{lm:M_sum-SPRT_per+opt} is finished.

\subsection{Proof of Theorem \ref{th:basic_char}}\label{ap:theo_prelimi}
Before we prove results in Theorem~\ref{th:basic_char}, we need to present two lemmas. The following paragraph is the shared assumption of those two lemmas.

Assume $\bm{x}(k)\triangleq[x_1(k),\dots,x_m(k)]$ is a sequence of independently and identically distributed random vectors of $m$ dimensions. Denote the probability measure and expectation with respect to it as $\Pb,\Eb$ and denote the expectation of every element as $\eta\triangleq\Eb[x_1(k)]$. We assume $0<\eta<\infty$. Define random walk $S_i(n)\triangleq\sum_{k=1}^{n}x_i(k)$.
	\begin{lemma}\label{lm:cross_time_secu}
	Define two stopping time in the following ($b>0$)
	$$\overline{T}(b)\triangleq\inf\{n\in\Zb^+, \max_{1\leq i\leq m}S_i(n)\geq b\},$$
	$$\underline{T}(b)\triangleq\inf\{n\in\Zb^+, \min_{1\leq i\leq m}S_i(n)\geq b\}.$$
	Then we have
	$$\lim_{b\ra\ift} \Eb\left|\frac{\overline{T}(b)}{b}-\frac{1}{\eta}\right|=0,$$
	$$\lim_{b\ra\ift}\Eb\left|\frac{\underline{T}(b)}{b}-\frac{1}{\eta}\right|=0.$$
\end{lemma}

\begin{proof}	
	It's a special case of \cite{farrell1964-1} Theorem 1.
\end{proof}

\begin{lemma}\label{lm:false_alarm_secu}
	We assume that there exist $h<0$ so that $\Eb[e^{h x_i}]=1$ and in addition $\Eb[x_i e^{h x_i}]<\ift$.
	The stopping time $\tau^-_i(a)$ is defined same as equation (\ref{eq:def_taua}).
	We have the following result for every $1\leq i\leq m$:
	$$\lim_{a=b\ra \ift}\frac{1}{a}\log{\Pb[\tau^-_i(a)<\ift]}=-|h|.$$
\end{lemma}

\begin{proof}
	The symmetric result ($\mu<0,h>0$) is provided in \cite{upper_bound_kugler2013} Section 1. As all $x_i$ are identically distributed, results hold for all $1\leq i\leq m$. 
\end{proof}

Now we can proceed to prove Theorem~\ref{th:basic_char}.

\textbf{Part (1)}
	\begin{equation*}
	\lim_{a=b\ra\ift}\Eb_0\left|\frac{\tau^-_{(r)}(a)}{a}-\frac{1}{I_0}\right|=0, \
	\lim_{a=b\ra\ift}\Eb_1\left|\frac{\tau^+_{(r)}(b)}{b}-\frac{1}{I_1}\right|=0. 	
	\end{equation*}
\begin{proof}
	We first prove the latter one and the former one can be dealt with similarly. 

	Define the first time when there are $r$ statistics $S_i(k)$ above threshold $b$ or below threshold $-a$ at the same time:
	\begin{align*}
	T^+_{(r)}(b)\triangleq &\inf_{k\in\Zb^+}\{S_{(s-r+1)}(k)\geq b\}, \\
	T^-_{(r)}(a)\triangleq &\inf_{k\in\Zb^+}\{S_{(r)}(k)\leq -a\},
	\end{align*}
	where $S_{(i)}(k)$ is the ascending ordered cumulative log-likelihood ratio $S_i(k)$, i.e., $S_{(1)}(k)\leq S_{(2)}(k)\leq\dots S_{(s)}(k)$.
	In the absence of attack, for the same $a$ or $b$ we have the following inequality:
	$$\tau^-_{(1)}(a)=T^-_{(1)}(a)\leq \tau^-_{(r)}(a)\leq\tau^-_{(s)}(a)\leq T^-_{(s)}(a),$$
	$$\tau^+_{(1)}(b)=T^+_{(1)}(b)\leq \tau^+_{(r)}(b)\leq\tau^+_{(s)}(b)\leq T^+_{(s)}(b).$$
	It suffices to prove 
	$$\lim_{a=b\ra\ift}\Eb_1\left|\frac{T^+_{(1)}(b)}{b}-\frac{1}{I_1}\right|=0,\ 	
	\lim_{a=b\ra\ift}\Eb_1\left|\frac{T^+_{(s)}(b)}{b}-\frac{1}{I_1}\right|=0.$$
	According to Lemma \ref{lm:cross_time_secu} (notations $\overline{T},\underline{T}$ are also from Lemma \ref{lm:cross_time_secu}),  $T^+_{(s)}(b)=\overline{T}(b)$, $T^+_{(1)}(b)=\underline{T}(b)$. If we set $m=s$ and those elements $x_i(k)$ in vector $\bm{x}(k)$ as log-likelihood ratios $L_i(k)$ from $s$ sensors, statement above is true because exception of log-likelihood ratio by definition equals to K--L divergence.
	Thus, the proof is finished.
\end{proof}

\textbf{Part (2)}
$$\lim_{a=b\ra\ift}\frac{\Eb_0[{T}^{(r)}]}{a}\leq\frac{1}{I_0},\ \lim_{a=b\ra\ift}\frac{\Eb_1[{T}^{(r)}]}{b}\leq\frac{1}{I_1} .$$
\begin{proof}
	We prove the second one and the proof for the first one is similar. According to definition in (\ref{eq:def_vote_rule}), the stopping time of detection rule satisfy
	\begin{equation*}
	T^{(r)}=min\{\tau^-_{(r)}(a),\tau^+_{(r)}(b)\}.
	\end{equation*}
	Therefore,
	$$\lim_{a=b\ra\ift}\frac{\Eb_1[{T}^{(r)}]}{b}\leq\lim_{a=b\ra\ift}\frac{\Eb_1[\tau^+_{(r)}(b)]}{b}=\frac{1}{I_1},$$
	where the equation comes from Part (1). The proof is completed.
%
\end{proof}

\textbf{Part (3)}
	$$\lim_{a=b\ra\ift} \frac{1}{b}\log \Pb_0[\tau^+_{(r)}(b)\leq\tau^-_{(r)}(a)]\leq -r,$$
	$$\lim_{a=b\ra\ift} \frac{1}{a}\log \Pb_1[\tau^-_{(r)}(a)\leq\tau^+_{(r)}(b)]\leq -r.$$
\begin{proof}
	We prove the second one and the first one can be proved similarly. As what we have done in Part (1), we set those elements $x_i(k)$ in vector $\bm{x}(k)$ in Lemma \ref{lm:false_alarm_secu} as log-likelihood ratios $L_i(k)$ from $s$ sensors.
	We obtain the following asymptotic result (\ref{eq:h=+-1}) from Lemma \ref{lm:false_alarm_secu} by showing that conditions in the lemma is satisfied by $h=-1$. 
	\begin{equation}\label{eq:h=+-1}
	\Pb_1[\tau^-_i(a)<\ift]\sim Ce^{-a}\quad \forall i\in\Sp,\ a=b\ra\ift,
	\end{equation}
	where $C$ is a constant term. 
	
	In order to prove (\ref{eq:h=+-1}), it suffices to show that $\Eb_1[e^{-L_i(k)}]=1$ and $\Eb_1[L_i(k)e^{-L_i(k)}]<\ift$ for every $i\in\Sp,k\in\Zb^+$. First we have
	$$\Eb_1[e^{-L_i(k)}]=\int_{\Rb} \left(\frac{{\rm d}\mu(x)}{{\rm d}\nu(x)}\right)^{-1} {\rm d} \mu(x)=\int_{\Rb}{\rm d}\nu(x)=1.$$
	For the second one,
	$$\Eb_1[L_i(k)e^{-L_i(k)}]=\int_{\Rb} \log\left(\frac{{\rm d}\mu(x)}{{\rm d}\nu(x)}\right)  {\rm d} \nu(x)=I_1<\ift. $$
	
	Considering the i.i.d. setting, (\ref{eq:h=+-1}) is proofed. Event $\{\tau^-_{(r)}(a)\leq \tau^+_{(r)}(b)\}$ implies that there exists a index set $\Rp\triangleq\{i_1,i_2,\dots,i_r\}\subseteq\Sp$ that for every $i$ in the set, event $\{\tau^-_i(a)<\ift\} $ occurs.
	Considering the independence of every cumulative log-likelihood ratio, we obtain
	\begin{align*}
	&\Pb_1[\tau^-_{(r)}(a)\leq \tau^+_{(r)}(b)] \\
	\leq&\bigcup_{|\Rp|=r,\Rp\subseteq\Sp}\ \Pb_1\left[\max_{i\in\Rp}\tau^-_i(a)<\ift\right]\\
	\leq &\bigcup_{|\Rp|=r,\Rp\subseteq\Sp}\ \prod_{i\in\Rp}\Pb_1[\tau^-_i(a)<\ift]\sim \binom{s}{r}Ce^{-ra}. 
	\end{align*}
	This directly leads to result of Part (3) as the logarithm of constant term will converge to zero when divided by $a$.
\end{proof}

\subsection{Proof of inequalities (\ref{eq:CT_un_H1}) to (\ref{eq:FA_un_H0})}\label{ap:four_ineq}
\begin{proof}
	We prove (\ref{eq:CT_un_H1})(\ref{eq:FA_un_H1}) in this section and (\ref{eq:CT_un_H0})(\ref{eq:FA_un_H0}) can be proved in the same way.
	In order to analyze stopping time $\tau^+_{(r)}(b)$, without loss of generality, we assume $S_i(k)$ has been ordered by subscript for some fixed $k$, i.e.
	$$S_1(k)\leq S_2(k)\leq\cdots\leq S_s(k).$$
	The worst case of stopping time $\tau^+_{(r)}(b)$ under attack is that the largest $c$ cumulative log-likelihood ratios $S_i(k)$ are assigned to be smaller than all other ones. If we denote cumulative log-likelihood ratios of compromised sensor as $S'_i(k)$, then in the worst case the largest $c$ sensors (index $s-c+1$ to $s$) is manipulated and assigned to be small enough so that
	$$\max_{s-c+1\leq i\leq s} S'_i(k)<S_1(k).$$
	In this condition, the stopping time is reached if other $r$ honest $S_i(k)$ are no smaller than threshold $b$, i.e.
	$$\tau^+_{(r)}(b)\leq \inf_{k}\{S_{s-r-c+1}(k)\geq b\},$$
	which implies (\ref{eq:CT_un_H1}).
	For error probability $\Pb^g_1[\tau^-_{(r)}(a)\leq \tau^+_{(r)}(b)]$, the worst case is that cumulative log-likelihood ratios from compromised senors satisfy $S_i(n)\leq-a, \forall i\in\Cp,\forall n\in\Zb^+$, which means the wrong decision is made as long as there are $r-c$ mistaken votes because there have been $c$ manipulated votes. It can be denoted as event $\{\tau^-_{(r-c)}(a)\leq \tau^+_{(r-c)}(b)\}$ occur in absence of attack and (\ref{eq:FA_un_H1}) is thus obtained.
\end{proof}

%

	
	
	\bibliographystyle{IEEEtran}
	\bibliography{ref_zishuo}

\end{document}